\documentclass[journal]{IEEEtran}

\usepackage{amsmath}

\let\labelindent\relax

\usepackage{tikz, pgfplots, pgf}
\usepackage{tikz}
\usetikzlibrary{calc, intersections, positioning}
\usetikzlibrary{positioning}
\usetikzlibrary{shapes.geometric, arrows}
\tikzstyle{bblock} = [rectangle, rounded corners, minimum width=0.5in, minimum height = 0.2in, text centered, draw=red!30, fill=red!5]
\tikzstyle{Bblock} = [rectangle, rounded corners, minimum width=1.65in, minimum height = 0.45in, text centered, draw=red!80, dashed]
\tikzstyle{arrow} = [thick, ->, >=stealth]
\tikzstyle{line} = [thick, -, >=stealth]
\tikzset{node distance = 0.9in and 0.2in}

\usepackage{amsthm}
\usepackage{amssymb}  
\usepackage{mathtools}
\usepackage[utf8]{inputenc}
\usepackage{comment}
\usepackage{algorithm,algpseudocode}
\algblock{Initialization}{EndInitialization}
\algnotext{EndInitialization}
\algblock{Output}{EndOutput}
\algnotext{EndOutput}

\usepackage{breqn}
\usepackage{graphicx}
\usepackage[version=4]{mhchem}
\usepackage{siunitx}
\usepackage{longtable,tabularx}
\usepackage{graphics} 
\usepackage{float}
\usepackage{epsfig} 
\usepackage{times} 
\usepackage{color}
\usepackage{pifont}
\usepackage{enumitem}
\usepackage{bm}
\usepackage{bbm}
\usepackage{soul}
\usepackage{balance}
\makeatletter
\newcommand{\multiline}[1]{%
  \begin{tabularx}{\dimexpr\linewidth-\ALG@thistlm}[t]{@{}X@{}}
    #1
  \end{tabularx}
}
\makeatother
\usepackage{hyperref}
\hypersetup{
    colorlinks=true,
    linkcolor=blue,
    filecolor=blue,      
    urlcolor=blue,
    pdftitle={Overleaf Example},
    pdfpagemode=FullScreen,
    }

\newcommand{\R}{\mathbb{R}}

\newtheorem{assumption}{Assumption}

\newtheorem{lemma}{Lemma}
\newtheorem{remark}{Remark}

\newtheorem{proposition}{Proposition}
\newtheorem{corollary}{Corollary}

\newtheoremstyle{examplestyle}  
  {\topsep}   
  {\topsep}   
  {\normalfont}
  {}          
  {\bfseries} 
  {.}         
  {.5em}      
  {}          
  
\theoremstyle{examplestyle} 
\newtheorem{example}{Example}[section] 

\IEEEoverridecommandlockouts                              

\title{\LARGE \bf
Monte Carlo Grid Dynamic Programming: Almost Sure Convergence and Probability Constraints}
\author{Mohammad S. Ramadan$^1$, Ahmad Al-Tawaha$^2$, Mohamed Shouman$^3$, Ahmed Atallah$^4$, Ming Jin$^2$  

\thanks{$1:$ Mathematics and Computer Science Division, Argonne National Laboratory, Lemont, IL 60439, USA,  {\tt\small mramadan@anl.gov.}, $2:$ Departement of Electrical and Computer Engineering, Virginia Tech, Blacksburg, VA. $3:$ Research center of the smart vehicle, Toyota Technological Institute (TTI), Nagoya, Aichi, Japan, $4:$ Department of Mechanical \&\ Aerospace Engineering, University of California, San Diego, La Jolla CA 92093-0411, USA.}}

\begin{document}
\maketitle
\thispagestyle{empty}
\pagestyle{empty}

\begin{abstract}
Dynamic Programming (DP) suffers from the well-known ``curse of dimensionality'', further exacerbated by the need to compute expectations over process noise in stochastic models. This paper presents a Monte Carlo-based sampling approach for the state space and an interpolation procedure for the resulting value function, dependent on the process noise density, in a "self-approximating" fashion, eliminating the need for ordering or set-membership tests. We provide proof of almost sure convergence for the value iteration (and consequently, policy iteration) procedure. The proposed meshless sampling and interpolation algorithm alleviates the burden of gridding the state space, traditionally required in DP, and avoids constructing a piecewise constant value function over a grid. Moreover, we demonstrate that the proposed interpolation procedure is well-suited for handling probabilistic constraints by sampling both infeasible and feasible regions. The curse of dimensionality cannot be avoided, however, this approach offers a practical framework for addressing lower-order stochastic nonlinear systems with probabilistic constraints, while eliminating the need for linear interpolations and set membership tests. Numerical examples are presented to further explain and illustrate the convenience of the proposed algorithms.
\end{abstract}

\section{Introduction}

Optimal control in stochastic environments underpins various disciplines, ranging from robotics \cite{brudigam2021stochastic} to financial engineering \cite{soner2004stochastic}. In these domains, making accurate decisions under uncertainty is not merely advantageous but imperative for ensuring operational efficacy and system safety. Dynamic Programming (DP) \cite{bellman2015applied}, a foundational framework for framing and solving optimal control problems, suffers from the curse of dimensionality, particularly when confronting stochastic systems with noise processes \cite[ch~25]{doucet2001sequential}. This computational challenge has historically limited the applicability of DP, necessitating innovative approaches to mitigate its impact.


Efforts to address these challenges have led to the development of various methods. Among these, the approximation of the DP algorithm through stochastic approximation procedures \cite{chau2014overview, soner2004stochastic, watkins1992q} has paved the way for Approximate Dynamic Programming (ADP) \cite{si2004handbook} and Reinforcement Learning (RL) \cite{recht2019tour}. These approaches approximate the underlying value functions or policies, thereby circumventing the exhaustive computation across the entire state space. Concurrently, Model Predictive Control (MPC) \cite{keerthi1988optimal} offers a pragmatic solution by solving a finite-horizon optimal control problem in an open-loop fashion at each timestep \cite{mayne2014model}.


Despite these advancements, DP retains its relevance both as a conceptual foundation for algorithm development \cite{bertsekas2019reinforcement} and as a direct solution method for specific problems, such as Hamilton-Jacobi Reachability analysis \cite{herbert2021scalable}. The value function derived from DP offers critical insights into the system's performance, guiding control design and decision-making towards more favorable system states \cite{bertsekas1995dynamic}. Moreover, incorporating probabilistic constraints enhances DP's applicability in risk-sensitive environments \cite{meyn2008control}, underscoring its importance in safety-critical decision-making. However, integrating probabilistic constraints into DP formulations poses significant computational challenges, as it requires considering the joint probability distribution of the system's state variables and ensuring that constraints are satisfied with a desired probability level \cite{chow2019lyapunov}. This added complexity further exacerbates the curse of dimensionality and limits the scalability of traditional DP techniques. Additionally, computational advances have extended the reach of optimized solvers \cite{bui2022optimizeddp}. This efficiency highlights DP's continued relevance in addressing lower-order stochastic nonlinear systems, which typically do not admit elegant optimal control solutions.


This paper addresses the inherent challenge of optimizing control strategies in stochastic environments with probabilistic constraints. We focus on developing a methodology that respects the operational constraints and navigates the computational complexities traditionally associated with DP. To this end, we present a Monte-Carlo-based sampling algorithm of the state space, with an interpolation procedure of the resulting value function and the corresponding feedback law that is dependent on the process noise density. In \cite{rust1997using}, a similar randomization and interpolation procedure is presented with its proof of convergence and upper bounds on its computational complexity. Different from \cite{rust1997using}, we relax the sampling restriction and allow for distributions other than uniform,  enabling adaptive gridding via importance sampling. This can be vital, especially when certain regions in the state space require more exploration relative to others. Furthermore, our derivation extends naturally to cover the case of probabilistic constraints.

The paper follows a path analogous to that of \cite{bertsekas1975convergence}, except that instead of refining the grid to yield a better estimate, our random approach relies on the Borel-Cantelli lemma to reach similar convergence results, but in an "almost sure" sense. Moreover, the convergence results in this paper are based on estimating a value function represented by a self-approximating interpolation scheme, as described by \cite{rust1997using}, not a piecewise constant over the grid as in \cite{bertsekas1975convergence}. We conclude the paper with numerical examples to further explain our algorithms and illustrate their practical convenience.

\section{Finite-Horizon}
 \label{Section:finitehorizon}
 Consider a finite-horizon optimal control problem aimed at minimizing the finite-horizon cost functional
 \begin{align*}
     J(x_0,u_{0:T-1})&=\textbf{E}_{x_{1:T}} \left \{ \sum_{k=1}^{T-1} \ell_k(x_k,u_k) +  \ell_N(x_T) \mid x_0,u_{0:T-1} \right \},
 \end{align*}
 over the space of control sequences $u_{0:T-1}= \{u_k,\,k\in \mathbb{I}_{T-1}\}$, for every initial state $x_0$, where the notation $\mathbb{I}_{T-1}=\{0,1, \ldots, T-1\}$. The expectation is taken over the state sequence $x_{1:T}$, which, its distribution is characterized by the dynamics\footnote{Our forthcoming derivations can be extended beyond the additive noise case.}:
\begin{align}
     x_{k+1}=f(x_k,u_k)+w_k,\quad x_0\in\mathbb{X}\subset\R^{r_x} \text{ is given,} \label{FDPE:Dynamical System}
 \end{align}
where $\{w_k\}_k$ is independent and identically distributed, each with a density $\mathcal{W}$. Furthermore, the control inputs are subject to constraints: $u_k \in U\left(x_k\right) \subseteq \mathbb{R}^{r_u},$ for all $x_k \in \mathbb X$, where
\begin{align}
U(x_k)=\left\{u_k \in \mathbb{U}\subset \mathbb{R}^{r_u} \mid \operatorname{Pr}\left(x_{k+1} \in \mathbb{X} \mid x_k, u_k\right) = 1 \right\}, \label{forward_invariance_constraint}
\end{align}
and $\mathbb U$ being the admissible control space. The $U(x_k)$ set is assumed non-empty for all $x_k \in \mathbb{X}$. This constraint ensures that the resulting closed-loop system renders $\mathbb X$ forward invariant, that is, $x_k \in \mathbb X$, for all $k$, starting from $x_0 \in \mathbb X$.

 The DP algorithm corresponding to this optimal control problem is
 \begin{equation}
 \begin{aligned}
     V_T(x_T)&=\ell_T(x_T),\, x_T \in \mathbb{X}_T\subset \R^{r_x},\\
     V_{k}(x_k)&= \min_{u_k \in U(x_k)} \{\ell_k(x_k,u_k) + \\
     &\hskip 20mm \textbf{E}_{x_{k+1}} [V_{k+1}(x_{k+1})\mid x_k,u_k]\},\label{FDPE}\\
     &x_k\in \mathbb{X},\,u_k\in U(x_k),\, k\in \mathbb{I}_T,
 \end{aligned}
 \end{equation}
\begin{assumption} \label{Assumption1}
\begin{enumerate}
    \item The set $\mathbb{X} \subset \R^{r_x}$ is compact, $\mathbb U$ is finite, and $U(x_k)$ is non-empty for all $x_k \in \mathbb X$.
    \item The functions $f(\cdot,u_k)$ and $\ell_k(\cdot,u_k)$ are Lipschitz almost everywhere (with the respect to the corresponding Lebesgue measure over their first argument) on $\mathbb{X}$, for all $u_k \in U(x_k)$.
    \item $\ell_k(x_k,u_k) \geq 0$ for all $x_k \in \mathbb{X}$ and $u_k \in U(x_k)$.  
\end{enumerate}
\end{assumption}
 
These assumptions are standard in dynamic programming, ensuring feasibility and stabilizability.

\vskip 3mm
\begin{remark}
The assumption of having $U(x_k)$ non-empty for all $x_k\in \mathbb X$, in Assumption~\ref{Assumption1}, which translates into the necessity of having $\mathcal{W}(\cdot-f(x_{k-1},u_{k-1}))\subset \mathbb{X}$ for all $x_{k-1}$ and some $u_{k-1}$, might not be possible in a general practical setting. For example, if the process noise is Gaussian. In this case, it might be required to reformulate the original optimal control problem, for instance, truncating and re-normalizing the densities and/or tightening the feasible set $\mathbb X$. These reformulations, existing in all grid-based DP algorithms, might introduce approximation errors, which from practical experience, can be mitigated if $\mathbb X$ is large enough.
\end{remark}

We now present our Monte Carlo based sampling algorithm to solve \eqref{FDPE}.

\begin{algorithm}[ht]
\caption{\label{algorithmDPE}Particle DP (Finite-Horizon)}
\begin{algorithmic}[0]
\Initialization

Let $\mathcal{X}$ be a density function on $\mathbb{X}$ such that $\textit{supp}(\mathcal{X}) = \mathbb{X}$, and let the particles $\{\xi^j\}_{j=1}^{N}$ be i.i.d. sample according to it. Then evaluate the weights $\{\Omega^l_T\}_{l=1}^{N}$
         \begin{align} \label{terminal_weights}
            \Omega^l_T=V_T(\xi^l)=\ell_T(\xi^l),\,l=1,\hdots,N.
        \end{align}
\EndInitialization

\For{$k=T-1$ down to $0$,}
    \begin{align}
        \Omega^l_{k}&=\min_{u_k \in U(\xi^l)}\{\ell_k(\xi^l,u_k)+\sum_{j=1}^{N} \Omega^j_{k+1} c^j(\xi^l,u_k)\},\label{FDPE_weights}\\
        &\hskip5mm l=1,2,\hdots,N, \nonumber
    \end{align}
        where 
    \begin{align}
     c^j(x_k,u_k)&=\frac{ \mathcal{M}(x_k,u_k)^j}{\sum_{i=1}^{N} \mathcal{M}(x_k,u_k)^i},
 \end{align}
    and 
    \begin{align*}
        \mathcal{M}(x_k,u_k)^j = \frac{\mathcal{W}(\xi^j-f(x_k,u_k))}{\mathcal{X}(\xi^j)}.
    \end{align*}
\EndFor
\end{algorithmic}
\end{algorithm}

We shall show next that Algorithm~\ref{algorithmDPE} can be used to generate an approximation $\tilde V_k^{N}$,
\begin{align}
    \tilde V_k^{N}(x_k)=\min_{u_k \in U(x_k)} \left \{\ell_k(x_k,u_k)+\sum_{j=1}^{N} \Omega^j_{k+1} c^j(x_k,u_k)\right \}, \label{FDPE:approx}
\end{align}
where $\Omega_{k}^{j}$ and the $c$ terms represent the ``value weights'' associated with the $j$th particle and the importance likelihoods, at time $k$, respectively. Their proper definition is established in Algorithm~\ref{algorithmDPE}. The approximation $\tilde V_k^{N}$ converges, almost surely, to $V_k$, for all $k=1,2,\hdots,T$. We first prove it for the terminal time-step and then use induction (backwards) for the remaining time-steps.

The control law (or policy) corresponding to this approximation of the value function can be found as follows:
\begin{align}
    u_k^*(x_k)=\operatorname*{arg\,min}_{u_k \in U(x_k)}\left \{\ell_k(x_k,u_k)+\sum_{j=1}^{N} \Omega^j_{k+1} c^j(x_k,u_k)\right\} \label{controller}
\end{align}

\begin{lemma}
\label{lemma:convergence T-1}
The approximation $\tilde V_{k}^{N}$ in (\ref{FDPE:approx}), when $k=T-1$, converges to $V_{T-1}$ almost surely on $\mathbb X$, as $N \to \infty$.
\end{lemma}
\begin{proof}
We start with the expectation term on the right-hand-side of the DP (\ref{FDPE}), when $k=T-1$,
\begin{align}
    &\textbf{E}_{x_{T}} \left \{V_{T}(x_{T})\mid x_{T-1},u_{T-1}\right \}=\nonumber\\
    &\hskip 5mm\int_{\mathbb{X}_{T}} V_{T}(x_{T})p(x_{T}\mid x_{T-1},u_{T-1})\,dx_{T}, \label{value_pred_expectation}
\end{align}
where $p(x_{T}\mid x_{T-1},u_{T-1})$ is the prediction density, and it can be expressed in terms of the process noise density $\mathcal{W}$ which is assumed given. In particular, the prediction density $$p(x_{T}\mid x_{T-1},u_{T-1}) = \mathcal{W}(x_{T}-f(x_{T-1},u_{T-1})).$$ 
Since $\{\xi^j\}$ is i.i.d. with density $\mathcal{X}$ supported by $\mathbb{X}$, if the support of the prediction density $\mathcal{W}(\cdot-f(x_{T-1},u_{T-1}))\subset \mathbb{X}$ for all $x_{T-1} \in \mathbb{X}$ and some $u_{T-1} \in U$, then the expectation \eqref{value_pred_expectation} can be expressed by a weighted average using ideas from Monte Carlo integration. 

The following sum
\begin{align*}
    \sum_{j=1}^N V_{T}(\xi^j) c^j(x_{T-1},u_{T-1})
\end{align*} 
is the self-normalizing importance sampling estimate \cite{liu2001monte} of the expectation \eqref{value_pred_expectation}, and using \eqref{terminal_weights},
\begin{align*}
    \sum_{j=1}^N V_{T}(\xi^j) c^j(x_{T-1},u_{T-1}) = \sum_{j=1}^N \Omega^j_T c^j(x_{T-1},u_{T-1}).
\end{align*} 
This estimate is unbiased and converges almost surely on $\mathbb X$ as $N\to\infty$ \cite{liu2001monte}. Therefore, 
\begin{align*}
    \ell_{T-1}(x_{T-1},u_{T-1})+\sum_{j=1}^{N} \Omega^j_{T} c^j(x_{T-1},u_{T-1})
\end{align*}
converges almost surely to 
\begin{align*}
    \ell_{T-1}(x_{T-1},u_{T-1}) + \textbf{E}_{x_{T}} [V_{T}(x_{T})\mid x_{T-1},u_{T-1}],
\end{align*}
as $N \to \infty$, for all $u_{T-1} \in U_{T-1}$. Since $U_{T-1}(x_{T-1})$ is finite, by Assumption~\ref{Assumption1}, both minimums over $U_{T-1}(\xi^l)$ are equal almost surely as well, as $N\to \infty$. Therefore, we have
\begin{align} \label{Convergence_omega_T}
\lim_{N\to\infty} \Omega^l_{T-1} = V_{T-1}(\xi^l).
\end{align}
\end{proof}

We now show that this results hold for the time-step $T-2$ as well, and then, by induction, it holds for time steps down to $0$.
\vskip 3mm
\begin{proposition} 
\label{proposition1}
The approximation $\tilde V_{T-2}^{N}$ generated by Algorithm~\ref{algorithmDPE}, in (\ref{FDPE:approx}), converges to $V_{T-2}$ almost surely over $\mathbb X$, as $N \to \infty$.
\end{proposition}
\begin{proof}
Suppose an infinite set of particles $\{\xi^j\}_j^\infty$ is generated as in Algorithm~\ref{algorithmDPE}. Let $N^+,\,N^-$ be positive integers. In the following, for some integer $l$,
\begin{align*}
     \ell_{T-2}(x_{T-2},u_{T-2})+\sum_{j=1}^{N^-} \Omega^j_{T-1} c^j(x_{T-2},u_{T-2}),
\end{align*}
suppose $\Omega^j_{T-1}$ is evaluated using Algorithm~\ref{algorithmDPE} by the first $N^+$ particles, that is, $\{\xi^j\}_j^{N^+}$. Using Lemma~\ref{lemma:convergence T-1}, in particular \eqref{Convergence_omega_T}, and the fact that $\Omega^j_{T-1}$ is a function of $N^+$,
\begin{align*}
     &\lim_{N^+\to\infty}\left\{\ell_{T-2}(x_{T-2},u_{T-2})+\sum_{j=1}^{N^-} \Omega^j_{T-1} c^j(x_{T-2},u_{T-2})\right\}\\
     &=\ell_{T-2}(x_{T-2},u_{T-2})+\sum_{j=1}^{N^-} \lim_{N^+\to\infty} \left \{ \Omega^j_{T-1} \right \} c^j(x_{T-2},u_{T-2}),\\
     &=\ell_{T-2}(x_{T-2},u_{T-2})+\sum_{j=1}^{N^-} V_{T-1}(\xi^j) c^j(x_{T-2},u_{T-2}).
\end{align*}
Applying the same convergence result due to importance sampling, as in Lemma~\ref{lemma:convergence T-1}, 
\begin{align*}
     &\lim_{N^-\to\infty} \Omega_{T-2}^l=\\
     &\lim_{N^-\to\infty}\left\{\ell_{T-2}(\xi^l,u_{T-2})+\sum_{j=1}^{N^-} V_{T-1}(\xi^j) c^j(\xi^l,u_{T-2})\right\}
     \\
     &=V_{T-2}(\xi^l), \text{ for all $l$}.
\end{align*}
The order of the limit above can be reversed, and one can then use the dominated convergence theorem and reach the same result.
\end{proof}

Next, by induction, we extend the convergence result and show that it holds for all time-steps, from $k=T-1$ and down to $k=0$.
\vskip 3mm
\begin{corollary}\label{corollary1}
The Monte Carlo approximation $\tilde V_{k}^{N}$, in (\ref{FDPE:approx}), converges to $V_{k}$ almost surely on $\mathbb X$, as $N \to \infty$.
\end{corollary}
\begin{proof}
By induction, using Lemma~\ref{lemma:convergence T-1} and Proposition~\ref{proposition1}.
\end{proof}

The results in the paper so far are exclusively valid for the case of $U(x_k)$ being finite. We seek to extend that for the case when $U(x_k)$ is compact. For this case, we need the Lipschitzness of $\ell(x_k,u_k)$ and $f(x_k,u_k)$ in $u_k$ as well. We prove the convergence over a randomly generated finite grid of $U(\xi^l)$ as this number of samples in the grid approach $\infty$. The Lipschitzness over $u_k$ implies the boundedness and convergence to $0$ of the error introduced the piecewise linear interpolation over this grid. We start with the following result.
\vskip 3mm
\begin{lemma}
(Proposition~1 in \cite{bertsekas1975convergence}) $V_k$ is Lipschitz on $\mathbb{X}$, for all $k\in \mathbb{I}_T$.
\end{lemma}
\begin{proof}
Provided in \cite{bertsekas1975convergence}, with the sum there replaced by the expectation in
\begin{align*}
    V_{T-1}(x_{T-1})&= \inf_{u_{T-1} \in \mathbb{U}_{T-1}} \textbf{E}_{x_{T}} \{\ell_{T-1}(x_{T-1},u_{T-1}) \\
    &\hskip 15mm+ V_{T}(x_{T})\mid x_{T-1},u_{T-1}\},\\
    &=\inf_{u_{T-1} \in \mathbb{U}_{T-1}} \{\ell_{T-1}(x_{T-1},u_{T-1}) +\\
    &\hskip 15mm\int_{\mathbb{X}_T}V_{T}(x_T)p(x_T\mid x_{T-1},u_{T-1})\,dx_T\},
\end{align*}
and then all the following steps therein hold true here as well.
\end{proof}

For a fixed $x_k \in \mathbb X$, let $\{\nu^q\}_{q=1}^{N_q}$ be i.i.d. samples from a density $\mathcal{U}$ with a support $U(x_k)$. We introduce the following partitioning of the input space $U(x_k)$ that will help us later in the proof of convergence. Let the sequence of sets $\{A^i\}_{i=1}^{N_q}$ be such that
\begin{align*}
    A^i=\{u \in U(x_k): \text{argmin}_{q\in\mathbb{I}_{N_q}} \lVert u-\nu^q\rVert  = i\},
\end{align*}
where $\lVert \cdot\rVert $ denotes the Euclidean norm. Then $\{A^i,\,i\in \mathbb{I}_{N_q}\}$ partitions $\mathbb{X}$. Next,
define the maximum diameter of this partition as:
\begin{align*}
    d_s=\max_{i\in \mathbb{I}_{N_q}}\, \sup_{u\in A^i} \mid \mid u-\nu^i\mid \mid 
\end{align*}

\begin{lemma}
\label{lemmaBC}
As $N_q \to \infty$, the maximum diameter $d_s \to 0$ with probability $1$.
\end{lemma}
\begin{proof}
For any $u \in U(x_k)$, fix $\epsilon>0$. Define the open ball of radius $\epsilon$ and center $u$ as $B_\epsilon(u)=\{\bar u\in U(x_k): \lVert u-\bar u\rVert <\epsilon\}$. Define the event $C_i(u)=[\nu_i \in B_\epsilon(u)]$. Since the particles $\{\nu^q\}_q$ are generated according to a density supported by $U(x_k)$, we have $\text{Pr }(C_j(u))=\text{constant}>0$. Hence, $\sum_{j=1}^{N_q}\text{Pr }(C_j(u))\to \infty$ as $N_q\to \infty$. By the Borel-Cantelli Zero-One Law \cite{resnick2019probability}, $\text{Pr }[C_j(u) \text{, infinitely often}]=1$, in particular, there exists some finite positive integer $\bar N$, such that we get $\nu^{\bar N}\in B_\epsilon(u)$. This holds for all $\epsilon>0$, therefore $d_s$ has to vanish.
\end{proof}

Now we are ready to extend Corollary~\ref{corollary1} to the case when $U(x_k)$ is compact.
\vskip 3mm
\begin{proposition}
If restricting the calculation of the DP equation to $U(x_k)=\{\nu^q : q=1,\hdots,N_q\}$ yields $\hat V^{N_q}_k(x_k)$, then it converges almost surely to $V_k(x_k)$ ($U(x_k)$ is compact), for all $x_k \in \mathbb{X}$, as $N_q \to \infty$.
\end{proposition}
\begin{proof}
So far we have shown that
\begin{align*}
    \lim_{N \to \infty} \left \{ \ell_k(x_k,u_k) + \sum_{j=1}^N V_{k+1}(\xi^j) c^j(x_k,u_k) \right \} =\\
    \ell_k(x_k,u_k) + \int_{\mathbb X} V_{k+1}(x_{k+1}) \mathcal{W}(x_{k+1} - f(x_k,u_k)) dx_{k+1}=\\
    \ell_k(x_k,u_k) + \int_{\mathbb X} V_{k+1}(f(x_k,u_k)+w_k) p(w_k) dw_k=: \mathcal{B}(u_k),
\end{align*}
the first equality holds pointwise on $u_k$. Therefore we assume $N$ large enough and consider the integral expression in the following proof.

Notice that $\mathcal{B}(u_k)$ is Lipschitz in $u_k$, since $\ell,\,f,\,V_{k+1}(f(x_k,u_k)+w_k),\,c$ are all Lipschitz in $u_k$, for any $w_k$, hence, it is Lipschitz on $U(x_k)$, since it is compact. Therefore, suppose $u_k ^\star$ is the minimizer, that is, $u_k^\star = \text{argmin}\mathcal{B}(u_k)$. Then, since $\{A^i\}_i$ partition $U(x_k)$, there is $m(N_q) \in \mathbb{I}_{N_q}$ such that $u_k^\star \in A^{m(N_q)}$ and
\begin{align*}
    \mid \mathcal{B}(u_k^\star) - \mathcal{B}(\nu^{m(N_q)})\mid &\leq \mathcal{A} \lVert u_k^\star - \nu^{m(N_q)} \rVert\\
    &\leq \mathcal{A} d_s \to 0, \quad \text{as} \quad N_q \to \infty,
\end{align*}
where $\mathcal{A}$ is the Lipschitz constant of $\mathcal{B}$ over $U(x_k)$. Hence, 
\begin{align*}
    \lim_{N_q \to \infty} \mathcal{B}(\nu^{m(N_q)}) = \mathcal{B}(u_k^\star).
\end{align*}
For the convergence of the argmax, for any $\epsilon>0$, the set $U_0=\{u \in \R^{r_u} \mid \mathcal{B}(u_k^\star)>\mathcal{B}(u_k^\star)-\epsilon\}$ is non-empty and open, due to the Lipschitzness of $\mathcal{B}$, and has a probability $>0$ under $\mathcal{U}$, since $\text{supp}(\mathcal{U})=U(x_k)$. Therefore, similar to the reasoning above, the Borel-Cantelli Zero-One law tells us that there exists, with probability $1$, a finite $N_q^'$, such that $\mathcal{B}(\nu^{N_q'}) \in U_0$. This holds for all $\epsilon>0$, and therefore, $\nu^{m(N_q)} \to \{u \mid \mathcal{B}(u)=\mathcal{B}(u_k^\star)$ as $N_q \to \infty$, with probability $1$.

\end{proof}
 
 \section{Infinite-Horizon Discounted Cost}\label{section5}
The Dynamic Programming Equation (DPE) for the infinite-horizon discounted cost \cite{bertsekas2012dynamic} of the system (\ref{FDPE:Dynamical System}) is defined by
\begin{align}
    V(x_k) &= \min_{u_k \in U(x_k)}\{\ell(x_k,u_k)+ \alpha\textbf{E}_{x_{k+1}} \{V(x_{k+1})\mid \, x_k,u_k\}\},\nonumber\\
    &=\min_{u_k \in U(x_k)}\{\ell(x_k,u_k)+ \nonumber\\
    &\hskip10mm\alpha \int_{\mathbb{X}} V(x_{k+1})p(x_{k+1}\mid x_k,u_k)\,dx_{k+1}\}, \label{IH_DPE}
    \end{align}
    with a compact set $\mathbb{X} \in \mathbb{R}^{r_x}$, $U(x_k)$, as defined in the finite-horizon case, is finite, and the discount factor $\alpha \in (0,1)$. The prediction density $p(x_{k+1}\mid x_k,u_k)$ 
\begin{align}
    p(x_{k+1}\mid x_k,u_k)=\mathcal{W}(x_{k+1}-f(x_k,u_k)), \label{processNoise}
\end{align}
where the evolution of the stationary process $\{x_k\}_{k \geq 0}$ is defined as in \eqref{FDPE:Dynamical System}.

Before investigating the convergence of the DP algorithm as $T\to \infty$, we need to establish the existence of the corresponding cost over the infinite-horizon. The condition
\begin{align}
    \lim_{T \to \infty} \textbf{E}_{w_0,\hdots,w_T} \left \{ \sum_{k=0}^T \alpha^k\ell(x_k,u_k) \right \} < \infty,
\end{align}
can be achieved, for instance, by assuming the stage-cost $\ell$ to be bounded over $\mathbb X$, for all $u_k$, or equivalently, to be Lipschitz over $\mathbb X$ (hence bounded, since $\mathbb X$ is compact). Therefore, the optimal value (the minimum of the above cost) $V$ exists and is finite. Let the stage cost $ell_k$ in the previous section be such that $\ell_k=\alpha^k \ell$, and follow the corresponding notation for $V^N_T$. We now need to show that $\tilde V^N_T \to V$ as $T,N \to \infty$. 

Notice that, if $\ell < \beta$, 
\begin{align*}
    | V(x) - V_T(x) | \leq \frac{\alpha^T \beta}{1-\alpha},
\end{align*}
for all $x \in \mathbb X$, as $V(x) - V_T(x)$ is upper bounded by a geometric series. Since $\alpha^T \to 0$ as $T \to \infty$, $V_T \to V$ uniformly on $\mathbb X$ as $T \to \infty$. For the almost sure convergence of $\tilde V^N_T $ to $V$,
\begin{align*}
    | V(x) - \tilde V^N_T(x) | &=   | V(x) -V_T(x) + V_T(x) - \tilde V^N_T(x) |, \\ 
    & \leq | V(x) -V_T(x)| + |V_T(x) - \tilde V^N_T(x) | \to 0
\end{align*}
as $T \to \infty$, as shown above, and as $N \to \infty$, per Corollary~\ref{corollary1}\footnote{The adaptation of this corollary to include the discount factor $\alpha$ is straightforward.}.
 
The following algorithm is the analogous to Algorithm~\ref{algorithmDPE}, but for the infinite-horizon discounted cost case.

\begin{algorithm}[ht]
\caption{\label{algorithmDPE_infty}Particle DP (Infinite-Horizon Discounted)}
\begin{algorithmic}[0]
\Initialization

Let $\mathcal{X}$ be a density function on $\mathbb{X}$ such that $\textit{supp}(\mathcal{X}) = \mathbb{X}$, and let the particles $\{\xi^j\}_{j=1}^{N}$ be i.i.d. sample according to it. Initialize the weights $\{\Omega^l\}_{l=1}^{N}$, as zeros, randomly, or maybe,
\begin{align} \label{terminal_weights}
\Omega^l=\ell(\xi^l),\,l=1,\hdots,N.
\end{align}
\EndInitialization

\While{Convergence criterion not met,}
\begin{itemize}
    \item Evaluate, for all $l=1,2,\hdots,N$,
\begin{align*}
        \Omega^l_{+}&=\min_{u_k \in U(\xi^l)}\{\ell(\xi^l,u_k)+\sum_{j=1}^{N} \Omega^j c^j(\xi^l,u_k)\},
\end{align*}
        where $U(\xi^l)$ and $c^j(x_k,u_k)$ are as defined in Algorithm~\ref{algorithmDPE}.

\item Test convergence criterion between $\{\Omega^l_{+}\}_l$ and $\{\Omega^l\}$, if not met, $\Omega^l \gets \Omega^l_{+}$, for all $l$.
\end{itemize}
\EndWhile
\end{algorithmic}
\end{algorithm}

Algorithm~\ref{algorithmDPE_infty} approximates the value function as 
\begin{align*}
    \tilde V^N(x)= \min_{u \in U(x)}\{\ell(x,u)+ \alpha.\sum_{j=1}^N \Omega^j c^j(x,u)\},
\end{align*}
while the resulting control law, analogous to \eqref{controller},
\begin{align*}
    u^*(x)=\text{argmin}_{u \in U(x)}\left \{\ell(x,u)+\sum_{j=1}^{N} \Omega^j c^j(x,u)\right\} \label{controller_infty}
\end{align*}

\section{Probabilistic Constraints}
An important aspect of applying DP to real-world problems involves ensuring system safety and satisfying some operation conditions. In stochastic settings, the probabilistic constraint formulation becomes a natural option to adopt. In this section, we show how the incorporation of probabilistic constraints in our DP approaches is straightforward, due to the importance sampling results we rely on.

We consider the optimal control problem introduced in Section~\ref{Section:finitehorizon}, but, in addition, with probabilistic constraints: the state is to remain within a predefined safe set $\mathbb{S} \subset \mathbb{X}$ with a high probability
\begin{equation*}
\text{Pr}(x_k \in \mathbb{S}) \geq 1 - \epsilon, \quad \text{for all } k,
\end{equation*}
where $\epsilon \in [0,1)$ is the violation probability, typically very small $\epsilon \ll 1$. To effectively integrate these probabilistic constraints into the DP framework, we adjust the control space to include only those actions that satisfy the probabilistic safety requirement from any given state $x_k$, that is, the permissible control set at state $x_k$
\begin{align}
    \bar U(x_k)&=\{u_k \in U(x_k)\mid \,\text{Pr}(x_{k+1} \in \mathbb S\mid x_k,u_k)\geq 1- \epsilon\},\label{inputProbConst}
\end{align}
where $U(x_k)$ is defined as in \eqref{forward_invariance_constraint}.

Our Monte Carlo approach can naturally be leveraged to calculate the safety violation probability. Given a control action $u_k$ and the current state $x_k$, the probability of transitioning into the safe set in the next step is given by
\begin{align*}
\text{Pr}(x_{k+1} \in \mathbb{S} \mid x_k, u_k) &= \int_{\mathbb{S}} p(x_{k+1} \mid x_k, u_k),dx_{k+1},\\
&= \int 1_{\mathbb S}\, p(x_{k+1} \mid x_k, u_k),dx_{k+1},\\
&= \textbf{E} \left \{ 1_{\mathbb S}\mid x_k,u_k\right \},
\end{align*}
where $1_{\mathbb S}$ is the indicator function over the safe set $\mathbb S$. Towards using importance sampling, let $\{\xi^j\}_{j=1}^{N}$ be a set of particles sampled from a density $\mathcal{X}$ supported over $\mathbb X$. Particles that fall outside $\mathbb{S}$ help identify regions of the state space that violate the safety constraint. The set of indices for these particles is denoted by $\mathbb{I}$, that is,
\begin{equation}
\mathbb{I} = \{ j \mid \xi^j \in \mathbb{S}^c \},
\end{equation}
and $\mathbb{S}^c$ (the unsafe set) is the complement of $\mathbb{S} \in \mathbb{X}$. The safety probability is then approximated by: 
\begin{align*}
\text{Pr}(x_{k+1} \in \mathbb{S} \mid x_k, u_k) &= 1 - \text{Pr}(x_{k+1} \in \mathbb{S}^c \mid x_k, u_k),\\
&=1-\textbf{E} \left \{ 1_{\mathbb S^c}\mid x_k,u_k\right \},\\
&=1-\sum_{j \in \mathbb{I}} c^j(x_k, u_k),
\end{align*}
where $c^j(x_k, u_k)$ is defined as in Algorithm~\ref{algorithmDPE}. This approximate is, again, the self-normalizing importance sampling estimate \cite{liu2001monte}. It is unbiased, and converges almost surely as $N \to \infty$.

Using the importance sampling approximation above, the input constraint set $\bar U(x_k)$ in \eqref{inputProbConst} can be approximated by
\begin{align*}
    U^N(x_k)&=\{u_k \in U(x_k)\mid \,1-\sum_{j\in \mathbb I} c^j(x_k,u_k)\geq 1-\epsilon\}.
\end{align*}

The adaption, of the two DP algorithms presented in the previous sections, such that they incorporate probabilistic constraints, is done through the following steps
\begin{enumerate}
    \item The input space $U(\xi^l)$ presented in either algorithm, is to be replaced by $U^N(\xi^l)$.
    \item The algorithm is applied to the particles with indices outside of $\mathbb I$.
    \item If a particle $\xi^l$ results in an empty $U^N(\xi^l)$, then $l$ is added to $\mathbb I$.
\end{enumerate}

\section{Numerical Examples}
Two examples are shown here. The first is a simple linear Gaussian model where the infinite-horizon discounted cost value function is known. The second is a two-state nonlinear system which must satisfy a probability constraint.
\begin{example}
We consider a linear Gaussian state-space model, pivotal for deriving the infinite-horizon discounted-cost Linear Quadratic Regulator (LQR) solution. The model is expressed as $x_{k+1} = 0.95x_k + u_k + w_k, w_k \sim \mathcal{N}(0,0.5)$, facilitating a direct comparison between theoretical value functions and empirical results obtained via the Particle Value Iteration algorithm. The DPE for the value function, $V(x_k)$, is as follows:
\begin{equation*}
    V(x_k) = \min_{u} \left\{ x_k^TQx_k + u_k^TRu_k + \alpha \mathbf{E}_{w_k} \{V(Fx_k + Bu_k + w_k)\} \right\},
\end{equation*}
which simplifies to $V(x_k)=x_k^T X x_k+q$, where $X$ is derived from the Algebraic Riccati Equation (ARE). The theoretical value function is found to be $V(x_k) = 1.460x_k^2 + 6.570$.

Using the Particle Value Iteration Algorithm \ref{algorithmDPE_infty}, with initial particles $\{\xi^j\}_{j=1}^{2000}$ distributed according to $\mathcal{N}(0,4)$, and 50 control actions sampled from $\mathcal{N}(0,1)$, we approximate the value function as:
\begin{equation*}
    V(x_k) = 1.410x_k^2 - 0.0023x_k + 5.52,
\end{equation*}
illustrating the algorithm's capacity for closely approximating optimal control strategies within stochastic environments. The particle-based approach enhances computational efficiency by avoiding grid-based ordering and interpolation, typical in conventional dynamic programming. This facilitates parallelized approximation across particles, yielding speedups and scaling efficiently for complex scenarios with probabilistic constraints.
\end{example}


\begin{example}
\label{example2}
 Consider a nonlinear system characterized by a two-dimensional state vector $x_k=\left(x_{1,k}, x_{2,k}\right)^T$ evolving under specific dynamics influenced by stochastic disturbances $w_k=$ $\left(w_{1, k}, w_{2, k}\right)^T$, following a Gaussian distribution with zero mean and covariance $0.3 \,\mathbb{I}_{2 \times 2}$. The system's dynamics are represented by:
\begin{align*}
x_{1, k+1} &= 0.9\, x_{1, k} + 0.2\,x_{2, k} + w_{1, k},\\
x_{2, k+1} &= -0.15\,x_{1, k} + 0.9\,x_{2, k} + 0.05\,x_{1, k}\, x_{2, k} + u_k + w_{2, k}.
\end{align*}
Suppose the L-shaped state set, $\mathbb{S}^c=[3,5]\times[-4,2] \cup [-2,5]\times [-7,-4]$, is the unsafe set within $\mathbb{X}$ that the system aims to avoid. To establish an optimal control law, denoted as $\kappa$, we consider the value function $V$. This function represents the optimal cost over an infinite horizon, with stage costs defined by $\ell_k\left(x_k, u_k\right)=x_k^T x_k+u_k^2$ and a discount factor $\gamma=0.9$. 

The Value Iteration process is implemented using the proposed approach in Algorithm \ref{algorithmDPE_infty}. This involves generating a set of 2000 uniformly distributed points across the state space $\mathbb{X}=[-10,10] \times[-5,15]$ and 50 control action points within $\mathbb{U}=[-3,3]$. The algorithm iterates towards optimizing $\kappa$, with convergence criteria set to a maximum absolute relative error below $5 \%$. The outcomes of this iterative process are visually represented in
Figure \ref{fig:controlLawColorMap}, where the colormap illustrates the distribution of $\kappa$ across the state space. \footnote{The results of this example can be reproduced using our open-source \textsc{Python} code: {\tt https://github.com/msramada/MC_DynamicProgramming}}

\begin{figure}
\centering 
\includegraphics[scale=0.55]{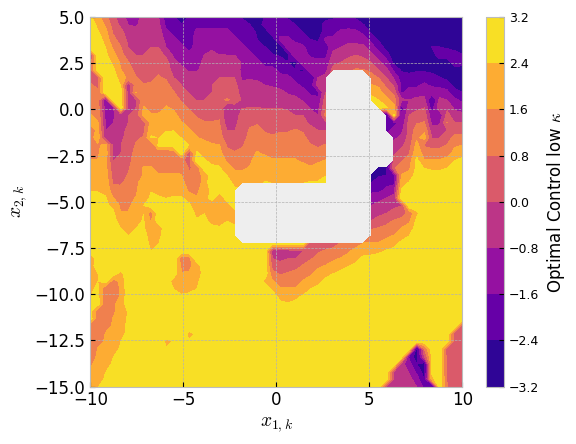}
\caption{Example \ref{example2}: Colormap of $\kappa$ acquired by Algorithm~\ref{algorithmDPE_infty} and fitted by a piece-wise surface. White regions correspond to infeasible states.}
\label{fig:controlLawColorMap}
\end{figure}
\end{example}

\section{Conclusion}
The convergence of the value function and feedback law approximations is implied under regularity conditions. The convergence rate depends on the dimensions of the state and input spaces and the complexity of the dynamics and distributions. Thus, this approach does not avoid the curse of dimensionality. However, the flexibility of the sampling process and the self-approximating forms mitigate some computational burden from implementing DP. This makes handling probabilistic constraints more natural. Higher sampling densities can be assigned to regions of greater importance, enabling adaptive gridding and more efficient resource use. This approach is a practical framework for lower-order stochastic nonlinear systems with probabilistic constraints.

\bibliographystyle{IEEEtran} 
\bibliography{References}
\end{document}